\documentclass[conference]{IEEEtran}
\IEEEoverridecommandlockouts
\usepackage{cite}
\usepackage{amsmath,amssymb,amsfonts}
\usepackage{algorithmic}
\usepackage{algorithm}
\usepackage{graphicx}
\usepackage{textcomp}
\usepackage{xcolor}

\usepackage{graphicx}
\usepackage{amsfonts}
\usepackage{amsmath}
\usepackage{amsthm}
\usepackage{amssymb}
\usepackage{mathabx}
\usepackage{ulem} 
\newtheorem{lemma}{Lemma}
\newtheorem{theorem}{Theorem}
\newtheorem{definition}{Definition}
\newtheorem{remark}{Remark}

\usepackage{xcolor}
\usepackage{hyperref}
\usepackage[left=1.63cm,right=1.63cm,top=1.63cm]{geometry}
\begin{document}

\title{Optimal Secure Coded Distributed Computation over all Fields\\
\thanks{Identify applicable funding agency here. If none, delete this.}
}

\author{\IEEEauthorblockN{Pedro Soto}
\IEEEauthorblockA{\textit{Department of Mathematics} \\
\textit{Virginia Tech}\\
Blacksburg, US \\
orcid.org/0000-0002-7120-7362}
}

\maketitle

\begin{abstract}
We construct optimal secure coded distributed schemes that extend the known optimal constructions over fields of characteristic 0 to all fields. 
A serendipitous result is that we can encode \emph{all} functions over finite fields with a recovery threshold proportional to the complexity (tensor rank or multiplicative); this is due to the well-known result that all functions over a finite field can be represented as multivariate polynomials (or symmetric tensors). 
We get that a tensor of order $\ell$ (or a multivariate polynomial of degree $\ell$) can be computed in the faulty network of $N$ nodes setting within a factor of $\ell$ and an additive term depending on the genus of a code with $N$ rational points and distance covering the number of faulty servers; in particular, we present a coding scheme for general matrix multiplication of two $m \times m $ matrices with a recovery threshold of $2 m^{\omega } -1+g$  where $\omega $ is the exponent of matrix multiplication which is optimal for coding schemes using AG codes. 
Moreover, we give sufficient conditions for which the Hadamard-Shur product of general linear codes gives a similar recovery threshold, which we call \textit{log-additive codes}. 
Finally, we show that evaluation codes with a \textit{curve degree} function (first defined in \cite{bensasetal}) that have well-behaved zero sets are log-additive. 
\end{abstract}


%
\IEEEpeerreviewmaketitle
\section{Introduction}
In this paper we consider the problem of coded distributed computation over a finite field.  
Coded distributed computing and, in particular, coded distributed matrix multiplication has attracted a large surge of research interest as of late~\cite{LLPPR2018, yma17, dfhjcg2020, YAA2018, DBJMG2018, tldk17,
ylrksa19, jj21, fc21, sma21, rk20, cgw21,jj21a, okk02024,schwartz2024}.
In this paper we will extend the batch matrix multiplication problem in \cite{ylrksa19, jj21}
to the case where there are more workers than there are elements in the field. 
We show that over finite fields, our rook codes can encode \emph{all functions.} We use codes constructed from algebraic function fields. Prior works that use algebraic geometry codes include \cite{fidalgodíaz2024distributed},  \cite{Okko}, \cite{HerA}, and evaluation codes \cite{SecureMatDot}.  

\subsection{Problem Statement}
We are given a function $T$ to evaluate on some $x$ that is inefficient to compute on one master node since 1) $T$ is of high complexity or 2) more simply, $x$ is too large, \emph{e.g.,} $x$ is a large dataset. 
The master node sends the other $N$ nodes in the network some $\tilde x_w$ such that the size $|\tilde x_w| << |x|$ and task the node $w \in [N]$ with computing\footnote{We could in theory consider the $\tau $ that each node $w$ computes to be a different $\tau_w$, but as we will see setting all of the $\tau _w = \tau$ is sufficiently robust to cover all cases of interest to an optimal complexity measure.} $\tau (\tilde x_w)$ where $\tau $ is a smaller computation than the function $T$. The main problem in coded distributed computing is the following: for a network of $N$ nodes find the minimum number of non-faulty nodes needed to return $\tau (\tilde x_w)$ so that we can recover the desired value $T(x)$. 
As is the convention, we consider only linear coding schemes since we need the encoding and decoding procedure at the master node to be as efficient as possible; in particular, we have that the function 
$\texttt{Enc} : (x_1,...,x_k) \mapsto (\tilde x_1,..., \tilde x_N)$
is given by a (full-rank) generator matrix $G_{\mathcal{C}} \in \mathbb{F}^{k \times N}$ so that $\tilde x = x G_\mathcal{C}$. We use the convention that this is an $[n,k]$ code when talking about codes in general and set $n=N$ when talking about nodes; the necessity of this convention will become clear in Section~\ref{sec:sec}. 
We assume the reader has basic familiarity with coding theory; \emph{e.g.,} such as that found in the introductory chapters of \cite{hp03} and \cite{v98}.

The preceding discussion is summarized by the following definition:
\begin{definition}
    The \textbf{recovery threshold} of a function $T$ in a network of $N$ nodes is the smallest number $\mathcal{R}(T,N)$ such that for any subset $S\subset [N]$ of size $\mathcal{R}(T,N)$ the value $T(x) = \texttt{Decode}(\{\tau(\tilde x_s) \}_{s \in S})$ can be recovered by some recovery procedure $\texttt{Decode}$.
    We call the number $N- \mathcal{R}(T,N)$ the \textbf{stragglers}; in classical coding theory, this is referred to as \textbf{erasures}.
    One can generalize this definition to include \textbf{byzantine} nodes as follows: 
    The recovery threshold tolerates $s$ stragglers and $b$ \textbf{byzantine} nodes if for any subset $S\subset [N]$ of size $\mathcal{R}(T,N)+2b$ the value $T(x) = \texttt{Decode}(\{\tau(\tilde x_s) + e_s \}_{s \in S})$ can be recovered if there is at most $b$ many $e_s \neq 0$. We extend the definition to colluding servers in Section~\ref{sec:sec}.
\end{definition}

\section{Hadamard-Shur Product}

 Given two linear codes $\mathcal{C},\mathcal{C}'$ with dimensions parameters $[n,k]$, we define their \textbf{Hadamard-Shur} (HS) product as 
 \begin{equation*}
   \mathcal{C} \circ \mathcal{C}' := \{(a_1b_1,...,a_nb_n) \in \mathbb{F} \mid \vec{a} \in \mathcal{C}, \vec{b} \in \mathcal{C}'\}.
 \end{equation*}
 It is easy to see that this defines another code with parameters $[k',n,d']$ for some $k'=:k(\mathcal{C} \circ \mathcal{C}')$ and $d'=:d(\mathcal{C} \circ \mathcal{C}')$ which we define as the dimension and distance of the resulting code.  
 The \textbf{designed recovery threshold} is defined as the value
 \begin{equation*}
 n- d(\mathcal{C} \circ \mathcal{C}')+1.
 \end{equation*}
 The fact that the designed recovery threshold corresponds to the recovery threshold 
We define the Hadamard-Shur powers as 
\begin{equation*}
    \mathcal{C}^{\circ d} := \mathcal{C} \circ ... \circ \mathcal{C}. 
\end{equation*}
 and the designed recovery threshold is defined as the value
 \begin{equation*}
 n- d(\mathcal{C}^{\circ d})+1 
 \end{equation*}
 in this case.

\begin{remark}
    Informally, the designed recovery threshold and the recovery threshold are equal for almost all of the schemes in this paper. 
\end{remark}
\section{Evaluation Codes and the HS Product}

An \textbf{evaluation code} (over $\mathbb{F}$) is given by the a finite set $D$ and a vector space of functions $V \subset  \mathbb{F}^D $ (where $\mathbb{F}^D$ is the set of all functions $\{f \mid f: D \rightarrow \mathbb{F}\}$), so that 
\begin{equation*}
\mathcal{C}(V,D) := \{ (f(P_1),...,f(P_n) ) \mid f\in V , P_i \in D\}.
\end{equation*}
If $\mathrm{dim}(V) = k$ and $|D| = n$, then $\mathcal{C}(V,D)$ is defined as a linear block code with parameters $[n,k]$ given by the generator matrix defined coordinate-wise as
\begin{equation*}
    (G_{\mathcal{C}(V,D)})_{i,j} = f_i(P_j)
\end{equation*}
where $V  = \left\langle f_1,...,f_k \right\rangle$ and $D = \{P_1,...,P_n\}$. 

Given two functions $f,g \in V$, we can define their product pointwise as 
\begin{equation*}
    (f \cdot g)(x) = f(x)g(x). 
\end{equation*}
We define the tensor product of two vector spaces of functions (with same domains) as 
\begin{equation*}
    V\otimes W := \{f\cdot g \mid f\in V \text{ and } g \in W\}. 
\end{equation*}
Since $V,W$ have the same domain, the evaluation code $\mathcal{C}(V\otimes W, D)$ is well defined. 

The proceeding discussion proves the following simplification for evaluation codes:
\begin{lemma}
We have the following equality between the Hadamard-Shur product and the tensor product of codes:
\begin{equation*}\label{lem:circ_to_ten}
    \mathcal{C}(V, D)\circ \mathcal{C}(W, D) = \mathcal{C}(V\otimes W, D).
\end{equation*}
Simlarly we have that 
\begin{equation*}
    \mathcal{C}(V, D)^{\circ d} = \mathcal{C}(V^{\otimes d} ,  D).
\end{equation*}
\end{lemma}

\section{HS Product of AG Codes}

We provide some background and definitions needed to construct AG codes. 
For a more complete proofs and definitions we refer the reader to \cite{Stichtenoth} and \cite{nied_xing}.
A \textit{function field}\footnote{In general, many authors allow the definition to include transcendence degree higher than 1, but all linear codes can be constructed using only transcendence degree 1, hence the limited scope of the definition.} over $\mathbb{F}$ is a field $\mathcal{F} \supset \mathbb{F}(x)$  that is an algebraic extension of $\mathbb{F}(x)$ (where $\mathbb{F}(x)$ is the field of formal rational functions over $\mathbb{F}$).  
The \textit{places}, denoted $\mathbb{P}_\mathcal{F} \subset 2^{\mathcal{F}}$, of a function field $\mathcal{F}$ is defined as the collection of all subsets, $P \subset \mathcal{F}$, that are a maximal ideal of a maximal local ring $\mathcal{O}_P \subset \mathcal{F}$.
Every such maximal ideal is principle, \emph{i.e.,} generated by a local parameter $P = t\mathcal{O}_P$, and thus, since $\mathcal{O}_P$ is a maximal local ring, we have that $(\forall x  \in \mathcal{F}) \  x \in \mathcal{O}_P \text{ or } x^{-1} \in \mathcal{O}_P$. 
Using these facts we can define a function $v_P: \mathcal{F} \rightarrow \mathbb{Z}$
\begin{equation*}
    v_{P}(x) = n \iff (\exists y \in \mathcal{O}_P \setminus t \mathcal{O}_P) \ x = t^n y  , 
\end{equation*}
and $v_P(0) = \infty$. 

Every element $f \in \mathcal{F}$ can be considered as a function via the \textit{residue class map} defined by
\begin{equation*}
f(P) := 
    \begin{cases}
        f \mod P & \text{ if } f \in \mathcal{O}_P \\ 
        \infty & \text{ otherwise}\\ 
    \end{cases}.
\end{equation*}
A place is called \textit{rational} if $\frac{\mathcal{O}_P}{P\mathcal{O}_P} \cong \mathbb{F}$, where $\frac{\mathcal{O}_P}{P\mathcal{O}_P}$ is the quotient field given by modding $\mathcal{O}_P$ by ${P\mathcal{O}_P}$.
More generally the \textit{degree} of a place, $P$, is the dimension of $\frac{\mathcal{O}_P}{P\mathcal{O}_P}$ when considered as a vector space over $\mathbb{F}$, usually denoted as $[\frac{\mathcal{O}_P}{P\mathcal{O}_P}:\mathbb{F}]$. We can extend the definition of the degree of a place to a divisor by linearity; \emph{i.e.,} $\mathrm{deg}(D) = \sum_{P \in \mathbb{P}_\mathcal{F}} c_P \mathrm{deg}(P)$. 

A \textit{divisor} is an element of the free abelian group generated by the places, \emph{i.e.,} $\mathbb{Z}^{\oplus \mathbb{P}_\mathcal{F}}$. We can order the divisors as follows $D_1 \leq D_2 \iff v_P(D_1) \leq v_P(D_2)$.
Each element of $\mathcal{F}$ has a \textit{principle divisor}, $(x) = \sum_{\mathcal{P \in \mathbb{F}_\mathcal{F}}} v_{P}(x)P$. We define the \textit{Riemann-Roch} of a divisor $G$ as the vector space of functions 
\begin{equation*}
   \mathcal{L}(G) := \{ f \in \mathcal{F} \mid (f) \geq - G \} \cup \{0\}.
\end{equation*}
For a divisor, we define $\ell(D)$ as the dimension of the Riemann-Roch space and we define the \textit{genus} of a function field as the value 
\begin{equation*}
    g := g(\mathcal{F}): =   \max_{ D \in \mathrm{div} (\mathcal{F}) } \{ \mathrm{deg} (D) - \ell (D) + 1  \}
\end{equation*}
We will see that the genus measures exactly how close the recovery threshold of the construction for finite fields can get to the characteristic 0 construction (\emph{e.g.,} Theorem~\ref{thm:main_mat_thm}; see Remark 6.1.5 of \cite{nied_xing}). 

\begin{definition}
    Given two divisors, $G,D$, of a function field, $\mathcal{F}$, where $D = P_1+...+P_n$ for some rational places $P_i$, we define an \textbf{algebraic geometry} code, denoted\footnote{This is the notational convention of \cite{Stichtenoth}; the authors of \cite{nied_xing} instead use the notation $C(P_1,...,P_n; G)$ for this construction.} $\mathcal{C}(G,D)$, as the evaluation code, $\mathcal{C}(V,D')$, where $V = \mathcal{L}(G)$ and $D' = D $.  
\end{definition}

In general, it is difficult to determine the dimension and distance of the HS product general evaluation code, without giving any information about $V$ and $D$; however, the dimension and distance of the HS product of a general Algebraic Geometry code is well known.  
In particular, we have the following: 

\begin{lemma}
We have the following equality between the Hadamard-Shur product and the tensor product of codes:
\begin{equation*}
    \mathcal{C}(G, D)\circ \mathcal{C}(H, D) = \mathcal{C}(G+H, D).
\end{equation*}
Similarly we have that 
\begin{equation*}
    \mathcal{C}(G, D)^{\circ d} = \mathcal{C}(d G ,  D).
\end{equation*}
\end{lemma}

\begin{lemma}
    The designed recovery threshold of the Hadamard-Shur product of two AG codes $\mathcal{C}(G,D),\mathcal{C}(H,D)$ is bounded above by 
\begin{equation*}
n - d (\mathcal{C}(G,D)\circ \mathcal{C}(H,D))+1  \leq   \mathrm{deg}(G+H)+1.
\end{equation*}
\end{lemma}

\begin{lemma}
    \label{lem:key_ag}
    If there exists a function field $\mathcal{F}$ with $n$ rational points and genus $g$, then there exists an HS code with designed recovery threshold
    \begin{equation*}
        n - d(\mathcal{C}(G,D)^{\circ \ell }) +1 = \ell k+g-1. 
    \end{equation*}
\end{lemma}
\begin{proof}
    Theorem 6.1.9 of \cite{nied_xing} gives the result. 
\end{proof}

\section{Log Additive Codes and the Curve Degrees of Evaluation Codes}
    Lemma~\ref{lem:key_ag} gives a recipe for an optimal coded distributed computing scheme which we formalize with the following definition: 
\begin{definition}
    We say that two linear codes $\mathcal{C}_1 , \mathcal{C}_2$ of length $n$ are \textbf{log additive} with respect to one another, denoted $\mathcal{C}_1 \sim \mathcal{C}_2$ if $ d(\mathcal{C}_1 \circ \mathcal{C}_2)\geq d(\mathcal{C}_1)+ d(\mathcal{C}_2) -n $. 
    We abuse notation and call a (single) code $\mathcal{C}$ \textbf{log additive} if $\mathcal{C} \sim \mathcal{C}$. 
    
\end{definition}

    Being log additive can also be characterized by the Hamming weight:
\begin{lemma}\label{lem:alt_log}
Let $w (x) $ be the Hamming weight of a code. Then two codes $\mathcal{C}_1 , \mathcal{C}_2$ are log additive if and only if 
\begin{equation*}
    w(x\circ y) \geq d(\mathcal{C}_1 ) + d(\mathcal{C}_2) - n 
\end{equation*}
for all $x \in \mathcal{C}_1, y \in \mathcal{C}_2$.
\end{lemma}
In order to characterize good distributed computing codes we will also need the notion of a genus for a general code: 
\begin{definition}
    The \textbf{(generalized\footnote{This definition is weaker than the stronger definition in \cite{vlanogtsa} which requires that the dual code have a similar bound.}) genus} of a $[n,k,d]$ code is  

    \begin{equation*}
        \mathfrak{g} = n-k+1-d. 
    \end{equation*}
    For AG codes it should be clear that $\mathfrak{g} \leq g$. 
    A \textbf{good distributed computing code for an $\ell$-linear tensor of rank $k$} is a log additive code $\mathcal{C}$ with parameters $[n , k, d]$ and with genus $\mathfrak{g}$ that satisfies the property 
    \begin{equation*}
        \ell \mathfrak{g} \leq  k. 
    \end{equation*}
    \end{definition}

    \begin{theorem}
        A log additive code has a designed recovery threshold equal to 
        \begin{equation*}
            n - d(\mathcal{C}^{\circ \ell}) + 1\leq   \ell k - 1 + \ell \mathfrak{g} .
        \end{equation*}
        In particular, a good distributed computing code has designed recovery threshold
        \begin{equation*}
            n - d(\mathcal{C}^{\circ \ell})+1 \leq   ( \ell +1)k - 1 .
        \end{equation*}
    \end{theorem}
    \begin{proof}
        We prove it for $\ell = 2$. 
        Notice that the designed recovery threshold
        \begin{multline*}
            n  - d(\mathcal{C}^{\circ 2})+1 \leq  2n - d(\mathcal{C}_1)- d(\mathcal{C}_2)  +1  \\ 
            = 2k + 2\mathfrak{g}  -2 +1   = 2k - 2 \mathfrak{g} - 1
        \end{multline*}
        and therefore if $ \mathfrak{g} = n-k+1-d$ we get a recovery threshold of $3k-1$. 
    \end{proof}
   The Hasse-Weil bound for AG codes shows that for a $\ell$-linear tensor, $T$, of rank (see Section~\ref{sec:tensor} for a definition of rank) $R(T) < \frac{|n-q-1|}{ 2 \ell \sqrt{q}}$ the genus becomes the dominant term in the recovery threshold. This motivates looking for log-additive codes of small genus. 
Inspired by the construction \cite{bensasetal} of degree lifted codes to construct new local decodable(/correctable) codes, we show that general evaluation codes with a \textbf{(generalized\footnote{We call it ``generalized'' because we do not need condition 4 in Defition 1.2 of \cite{bensasetal}.}) curve degree function} are log additive. 
\begin{definition}\label{def:curve_deg}
A function
    $\mathfrak{d}: \mathcal{C}(V,D) \rightarrow \mathbb{Z}$ is a \textbf{(generalized) curve degree} if 
    \begin{enumerate}
        \item $\mathfrak{d}(f · g) = \mathfrak{d}(f) + \mathfrak{d}(g)$.
\item If $\mathfrak{d}(f) = d $, f either vanishes on at most $d$ points in $D$, or else it vanishes on all of $D$; \emph{i.e.,} $w(f) \geq n-d = |D|-d$ if $f \neq 0$.
\item  For any $\alpha,\beta \in  \mathbb{F} $ and any $ f, g \in V$, $\mathfrak{d}(\alpha · f + \beta · g) \leq \max\{\mathfrak{d}(f), \mathfrak{d}(g)\}$. Thus, the set
of functions $f \in V$ with $\mathfrak{d}(f) \leq  d$ is an $ F$-linear subspace which we denote as $\mathcal{C}(V,D)_d$.
    \end{enumerate}
\end{definition}

\begin{lemma}
    If $\mathcal{C}(V,D)$ has a curve degree $\mathfrak{d}$
    and there exists an $f \in V$ with $b$ many zeros in $D$, then the code 
   $\mathcal{C}(V,D)_b$
    is log additive. 
\end{lemma}

\begin{proof}
    By Lemma~\ref{lem:circ_to_ten} and the first condition in Definition~\ref{def:curve_deg}, we have that $(\mathcal{C}(V,D)_b)^{\circ 2} = \mathcal{C}(V^{\otimes 2},D)_{2b}$.
    The second condition in Definition~\ref{def:curve_deg} and our assumption on the number of zeros gives us that $w(f_1\cdot f_2) \geq n - 2b = 2(d(\mathcal{C}(V,D)_b) )- n $ for any 
    $ f_i \in    \mathcal{C}(V,D)_{b} $. 
    If we choose the minimal weight $g = f_1f_2$ we get $w(f_1\cdot f_2) = d(\mathcal{C}(V,D)_{2b})\geq n - 2b =2 d(\mathcal{C}(V,D)_b) -n $ and Lemma~\ref{lem:alt_log} completes the proof.
\end{proof}

\section{Extending Codes to Vectors}
It is straightforward to prove that any $[n,k,d]_q$-code can be considered as a $[n,k,d]_{q^p}$ code with minimal overhead; in particular, 
\begin{lemma}
    Given an $[n,k,d]_q$ code, one can construct an $[n,k,d]_{q^p}$ with encoding and decoding complexity bounded by at most $p$ times the complexity for $[n,k,d]_q$. 
    More generally, for any algebraic extension $\mathbb{K} \supset \mathbb{F}$ of \emph{any} field $\mathbb{F}$ such that $\mathrm{dim}_\mathbb{F}(\mathbb{K}) = p$, we have that given an $[n,k,d]_\mathbb{F}$ code, one can construct an $[n,k,d]_{\mathbb{K}}$ with encoding and decoding complexity bounded by at most $p$ times the complexity for $[n,k,d]_\mathbb{F}$. 
\end{lemma}
\begin{proof}
    Given some $m \in \mathbb{F}_{q^p}^{k}$, consider is instead as an $p \times k$ matrix in $\mathbb{F}^{p \times k}$ and encode each of the $p$ rows individually. The proof is similar for decoding. 
\end{proof}
\begin{remark}
    The generality is to include characteristic 0 cases like the real and complex fields, where the overhead is at most doubled.  
\end{remark}

\section{Tensors}\label{sec:tensor}
An $\ell$-\textbf{linear tensor}\footnote{It is also common convention to call this form a $(\ell+1)$-tensor if $p>1$. For simplicity we only consider the order of the input and not the output, the more general order output form will be considered in the journal version.} is function 
$T:\mathbb{F}^{m_1} \times ...  \times \mathbb{F}^{m_\ell} \rightarrow \mathbb{F}^p$
which is linear in all of its arguments; \emph{i.e.,} if we fix the values of $x_1,...,x_{i-1},x_{i+1},...,x_{\ell}$, then 
\begin{multline*}
   T(x_1,...,\alpha x_i+ \beta z,..., x_\ell) 
   \\ = 
   \alpha T(x_1,...,x_i,..., x_\ell)  + 
   \beta T(x_1,...,z,..., x_\ell)  , 
\end{multline*}
for all $\alpha ,\beta \in \mathbb{F}$ and  $x_i,z \in \mathbb{F}^{m_i}$.
A tensor $t:\mathbb{F}^{m_1} \times ...  \times \mathbb{F}^{m_\ell} \rightarrow \mathbb{F}^p$ is called a \textbf{rank 1 tensor} if there exists some linear forms $t_i:\mathbb{F}^{m_i} \rightarrow \mathbb{F}$ and a $y \in \mathbb{F}^p$ such that 
\begin{equation*}
    t(x_1,...,x_\ell ) = \left(\prod_{i \in [\ell]} t_i(x_i) \right)y.  
\end{equation*}
The \textbf{tensor rank} of a tensor, $T$, is the smallest number, $R(T)$, such that there exists rank-one tensors $t^{(1)},...,t^{(R(T))}$ such that 
\begin{equation*}
    T(x_1,...,x_\ell) = \sum_{i \in [R(T)]} t^{(i)}(x_1,....,x_\ell).
\end{equation*}

The definition of rank 1 decomposition can be further generalized to a recursive decomposition as follows: given a tensor 
$T:\mathbb{F}^{m_1} \times ...  \times \mathbb{F}^{m_\ell} \rightarrow \mathbb{F}^p$
and a tensor 
$
\tau 
  :\mathbb{F}^{\mu _1} \times ...  \times \mathbb{F}^{\mu _\ell} \rightarrow \mathbb{F}^{\rho }
$
such that $\mu_i < m_i$ and $\rho < p$, we say that $T$ has \textbf{$\tau$-rank} $R_\tau(T)$ if there exists linear functions $\mathcal{E}_i^{(j)}: \mathbb{F}^{m_i}\rightarrow \mathbb{F}^{\mu_i}$ and some linear functions 
$
\mathcal{D}^{(j)}: \mathbb{F}^{\rho} \rightarrow \mathbb{F}^p 
$, where $i \in [\ell]$ and $j \in [R_\tau(T)]$,
such that 
\begin{equation*}
    T(x_1,...,x_\ell) = \sum_{j \in [R_\tau(T)]} \mathcal D^{(j)}(\tau(\mathcal E_1^{(j)}(x_1),...,\mathcal E_\ell^{(j)}(x_\ell))). 
\end{equation*}

\section{Fault Tolerant Matrix Multiplication}

We consider the following problem: given $k $ pairs of matrices 
$$
A_1,B_1,\dots,A_k,B_k,
$$
where  $A_i \in \mathbb{F}^{P \times S}$ and $B_i \in \mathbb{F}^{S \times Q}$ for $i \in [k]$ and a field $\mathbb{F}$, 
compute the products 
$$
A_1 B_1 ,\dots,A_k B_k
$$
in the distributed master worker topology in which the master node  gives  $N$ worker nodes coded matrices of the form 
$$
\tilde A _w = \sum_{i \in [k]} \alpha_{i}^{(w)} A_i  
\in \mathbb{F}^{P \times S}
, \tilde B _w = \sum_{i \in [k]} \beta_{i}^{(w)} B_i
 \in \mathbb{F}^{S \times Q}
$$
where $\alpha_{i}^{(w)} \in \mathbb{F}$ and $w \in [N]$ indexes over the worker nodes. 
We are concerned with the minimum number of worker nodes that need to return their values so that the master can recover the desired products. This will be formalized in Definition~\ref{def:rec_bat}.
Before we move on to constructing the actual codes, we will show that this batch matrix problem is actually the most general form of the distributed matrix multiplication problem since it implicitly solves the \textbf{general matrix-matrix} multiplication problem.
\subsection{General Matrix-Matrix Multiplication}\label{sec:mat_mult_explaination}

Given two matrices
\begin{equation}\label{eq:mat}
    A = \begin{bmatrix}
        A_{1,1} & \dots & A_{1,\zeta}\\
         \vdots & \ddots & \vdots \\
        A_{\chi,1} & \dots & A_{\chi,\zeta}\\
    \end{bmatrix} 
    , \ \ 
    B = \begin{bmatrix}
        B_{1,1} & \dots & B_{1,\upsilon}\\
         \vdots & \ddots & \vdots \\
        B_{\zeta,1} & \dots & B_{\zeta,\upsilon}\\
    \end{bmatrix}, 
\end{equation} with
$A_{i,j} \in \mathbb{F}^{P \times S}$ and $B_{i,j} \in \mathbb{F}^{S \times Q}$, we can take an optimal $(\chi, \upsilon, \zeta)$ fast matrix multiplication tensor, denoted $(\gamma, \delta, \eta)$, of rank $r:=R(\chi, \nu , \zeta )$, {i.e.,} for $t \in [r]$, defined by the equations 
\begin{align}\label{eq:fast_mat}
\begin{split}
\hat A _t = \sum_{i,j \in [\chi] \times [\zeta]} \gamma_{i,j}^{(t)} A_{i,j} , \ \ \hat B _t = \sum_{i,j \in [\zeta ]\times [\upsilon ]} \delta_{i,j}^{(t)} B_{i,j},
\\
\sum_{k \in [\zeta ]}A_{i,k}B_{k,j} = \sum_{t \in [r]} \eta^{(i,j)}_{t} \hat A_{t} \hat B_t ,
\end{split}
\end{align}
with
    and then  $A_i,B_i := \hat{A}_i,\hat{B}_i $ as in \cite{yu_ali_ave_2020}, which shows that the general matrix-matrix multiplication problem reduces exactly to the batch matrix multiplication. It seems that computing the linear functions defined by $\gamma,\delta,$ and $\eta$ is an undesirable overhead, \textit{but one must compute the functions given by $\alpha,\beta$ anyways; in particular, one may compose the $\alpha,\beta$ and the $\gamma, \delta$ so that there is no overhead.} This is made explicit and formalized in Section~\ref{sec:mat_mult}. 
    The number $r$ above is often called the \textit{tensor-rank} or the \textit{bilinear complexity}.

The primary goal of this paper is to generalize the main result of \cite{yu_ali_ave_2020} (which only holds in true generality generally for characteristic zero) to more general cases over finite fields. 
This paper 
overcomes the major obstacle when the field is finite, namely the number of evaluation points (meaning the number of workers) or more generally, the length of the code.  

   \subsection{General Construction for Matrix Multiplication}\label{sec:mat_mult} 

Here we construct the Dual Hadamard-Shur Product Code for Matrix Multiplication.
The construction for the product of the matrices given in Equation~\ref{eq:fast_mat} is given Algorithm~\ref{alg:cod_mat_mult}. 
 \begin{algorithm}
 \caption{Coded Matrix Multiplication}\label{alg:cod_mat_mult}
 \begin{algorithmic}[1]
 \renewcommand{\algorithmicrequire}{\textbf{Input:}}
 \renewcommand{\algorithmicensure}{\textbf{Output:}}
    \REQUIRE
Two codes $\mathcal{C}_1,\mathcal{C}_2$ with parameters $[n,k]$, where $n = N$, $k = r$ is equal to the tensor rank of matrix multiplication of type $(\chi,\zeta, \nu )$, and 
matrices $A,B$ as in Equation~\ref{eq:mat}
\STATE Let $H_1,H_2$ be the parity check matrices of the input codes
\STATE The master node uses Equation~\ref{eq:fast_mat} to compute $A_i : = \hat A_i, B_i:= \hat{B}_i$
\STATE The master node uses the parity check matrices to solve for two codewords $\tilde A, \tilde B$ in $\mathcal{C}_1,\mathcal{C}_2$ respectively which satisfy 
\begin{equation*}
    \tilde A_{n-k+i} = A_i \  \ \tilde B_{n-k+i} = B_i
\end{equation*}
\STATE  The node $w$ computes $\tilde A_w \cdot \tilde B_w$
\STATE The master node waits for $n - d(\mathcal{C}_1 \circ \mathcal{C}_2)+2b$ nodes to return 
\STATE The master node decodes the codeword in $\mathcal{C}_1 \circ \mathcal{C}_2$ treating the $d(\mathcal{C}_1 \circ \mathcal{C}_2)- 2b$ missing node's tasks as erasures 
to get 
\begin{equation*}
     \tilde C_{n-k+i} := \tilde A_{n-k+i}\cdot \tilde B_{n-k+i} = A_iB_i 
\end{equation*}
\STATE The master node uses $\eta$, the last component of the tensor in Equation~\ref{eq:fast_mat}, to solve for $C_{i,j} = \sum_{\kappa \in [\zeta] } A_{i,\kappa } B_{\kappa,j}$
 \end{algorithmic}
 \end{algorithm}

\begin{theorem}\label{thm:mat_mult_correct}
   If $\mathcal{C}_1,\mathcal{C}_2$ are two linear codes with parameters $[n,k]$
   which satisfy the conditions
   \begin{equation*}
       d(\mathcal{C}_1 \circ \mathcal{C}_2) \geq 2b +s +1,
   \end{equation*}
   \begin{equation*}
       k \geq R(\chi, \nu, \zeta ),
   \end{equation*}
   and 
   $
       n \geq N,
   $
   then the scheme given at Algorithm~\ref{alg:cod_mat_mult} can correctly compute the product $A\cdot B$ if there are at most $s$ straggler and $b$ byzantine nodes.
\end{theorem}

\begin{theorem}\label{thm:main_mat_thm}
    In the case where $\mathcal{C}_1 = \mathcal{C}_2 = \mathcal{C}(G,D)$ where $G,D$ are chosen from a function field with genus $g$ and $T$ is the matrix multiplication tensor corresponding to $(\chi,\nu, \zeta )$, we have that 
    \begin{equation*}
       \mathcal{R}(T,N) = 2 R(\chi, \nu, \zeta)-1+g.
    \end{equation*}
    In the particular case where $\chi = \nu = \zeta = m$ for some $m$ we have that 
    \begin{equation*}
       \mathcal{R}(T,N) = 2 m^{\omega }-1+g;
    \end{equation*}
    \emph{i.e.,}
    our scheme performs coded matrix multiplication up to a factor of 2 and additive term $g$ of the theoretically optimal complexity.
    
\end{theorem}

\section{Fault Tolerant Tensor Computation}

We now show how
given an $\ell$-linear tensor $T:\mathbb{F}^{m_1} \times ...  \times \mathbb{F}^{m_\ell} \rightarrow \mathbb{F}^p$ with tensor rank $r$, \emph{i.e.,  } 
\begin{equation*}
    T(x_1,...,x_\ell) = \sum_{i \in [r]} t^{(i)}(x_1,...,x_\ell),
\end{equation*}
and some $\ell$ codes $\mathcal{C}_1,...,\mathcal{C}_{\ell}$ with parameters $[n,r]$ where $n = N$ and $k = r$, we compute $T$ in a distributed fault tolerant manner using Algorithm~\ref{alg:cod_tensors}. 
 \begin{algorithm}
 \caption{Coded Tensor from Rank-1 Decomposition}\label{alg:cod_tensors}
 \begin{algorithmic}[1]
 \renewcommand{\algorithmicrequire}{\textbf{Input:}}
 \renewcommand{\algorithmicensure}{\textbf{Output:}}
    \REQUIRE
$\ell +1 $ codes $\mathcal{C}_1,...,\mathcal{C}_\ell,\mathcal{C}_{\ell +1}$ with parameters $[n,k]$, where $n := N$, $k := r$ is equal to the tensor rank of $T$, and $t^{(1)},...,t^{(r)}$ are the corresponding rank one tensors  
\STATE Let $H_1,...,H_{\ell},H_{\ell +1 }$ be the parity check matrices of the input codes
\STATE The master node uses Equation~\ref{eq:fast_mat} to compute 
$$
 t^{(i)}_j(x_1,...,x_\ell)
$$
\STATE The master node uses the parity check matrices to solve for $\ell$ codewords $\tilde t^{(1)} \in \mathcal{C}_1,...,\tilde t^{(\ell)} \in \mathcal{C}_\ell $ which satisfy 
\begin{equation*}
    \tilde t_{j}^{(n-k+i)}: = t^{(i)}_j(x_1,...,x_\ell)
\end{equation*}
and $\tilde y \in \mathcal{C}_{\ell + 1}$
such that 
\begin{equation*}
    \tilde y_{n-k+i} := \vec{y}_i
\end{equation*}
\STATE  The node $w$ computes $\tilde t^{(w)}_{1} \cdot ... \cdot \tilde t^{(w)}_{\ell } \cdot \tilde y_{w} $
\STATE The master node waits for $n - d(\mathcal{C}_1 \circ ... \circ  \mathcal{C}_\ell \circ \mathcal{C}_{\ell + 1} )+2b$ nodes to return 
\STATE The master node decodes the codeword in $\mathcal{C}_1 \circ ... \circ  \mathcal{C}_\ell \circ \mathcal{C}_{\ell+1} $ treating the $d(\mathcal{C}_1 \circ ... \circ \mathcal{C}_\ell\circ \mathcal{C}_{\ell+1})- 2b$ missing node's tasks as erasures 
to get 
\begin{multline*}
     \tilde C^{(n-k+i)} := \tilde t^{{n-k-i}}_{1} \cdot ... \cdot \tilde t^{(n-k-i)}_{\ell} \cdot \tilde y_{n-k+i} \\ 
     = \left(\prod_{j \in [\ell]} t^{(i)}_j(x_1,...,x_\ell)\right)\vec{y}_i = t^{(i)}(x_1,...,x_\ell)
\end{multline*}
\STATE 
The master node 
computes 
\begin{multline*}
    T(x_1,...,x_\ell) = \sum_{i \in [r]} \left(\prod_{j \in [\ell]} t^{(j)}_i(x_1,...,x_\ell)\right)\vec{y}_i 
    \\=  \sum_{i \in [r]} t^{(i)}(x_1,...,x_\ell)
\end{multline*}
 \end{algorithmic}
 \end{algorithm}

\begin{theorem}\label{thm:mat_mult_correct}
   If $\mathcal{C}_1,...,\mathcal{C}_{\ell+1}$ are two linear codes with parameters $[n,k]$
   which satisfy the conditions
   \begin{equation*}
       d(\mathcal{C}_1 \circ ... \circ  \mathcal{C}_\ell \circ \mathcal{C}_{\ell + 1} ) \geq 2b +s +1,
   \end{equation*}
   \begin{equation*}
       k \geq R(T ),
   \end{equation*}
   and 
   $
       n \geq N,
   $
   then the scheme given at Algorithm~\ref{alg:cod_mat_mult} can correctly compute the tensor $T(x_1,...,x_\ell)$ if there are at most $s$ straggler and $b$ byzantine nodes.
   
\end{theorem}

\section{Security}\label{sec:sec}

\begin{definition}
    A scheme is secure to $c$ colluding servers if $c$ servers cannot infer the inputs even if they share information with one another. 
\end{definition}
The informal construction is as follows: 
\begin{enumerate}
    \item We begin with a $[n,k]$ code $\mathcal{C}$ with $n = N+k+t$ where $k= R(T_\mathrm{mat\_mult})$ is the number of input matrices and $t$ is an input parameter 
    We generate $t$ random matrices with the same size as the $A_i$ and $t$ random matrices with the same size as the $B_i$ 
    \item We construct the scheme just as in Algorithm~\ref{alg:cod_mat_mult} but creating the extra $t$ matrices as inputs
    \item One punctures the two codes $\tilde A,\tilde B$ to delete the last $k+t$ points 
    \item one proceeds with the algorithm as usual 
\end{enumerate}
If one carries out the proof of security exactly as in Section of \cite{ylrksa19}, the only obstruction is that our code is not an MDS code. 
Therefore, instead of $t$ bits of random padding, one instead has $t- g$ bits of random padding (where $g$ is the genus). 
Therefore the construction is now secure to at most $b = t - \mathfrak{g}$ colluding servers. 
Furthermore the recovery threshold now takes the form 
\begin{equation*}
    \mathcal{R}(T,N) \leq 2( m^{\omega } +t) -1 + g= 2( m^{\omega } +b) -1 + 3g. 
\end{equation*}
\bibliographystyle{IEEEtran}
\bibliography{bib, skeleton}
%



\end{document}